\documentclass[11pt,letterpaper]{article}
\usepackage[margin=1in]{geometry}

\usepackage{times}
\usepackage{soul}
\usepackage{url}
\usepackage[hidelinks]{hyperref}
\usepackage[utf8]{inputenc}
\usepackage[small]{caption}
\usepackage{graphicx}
\usepackage{amsmath}
\usepackage{amsthm}
\usepackage{booktabs}
\usepackage{algorithm}
\usepackage{algorithmic}
\usepackage[switch]{lineno}
\usepackage{color}
\usepackage{xspace}
\usepackage{natbib}
\setcitestyle{authoryear,open={[},close={]}}
\usepackage{amsfonts,bbm,bm,mathtools,comment}
\usepackage{cleveref} 
\usepackage[normalem]{ulem} 

\newtheorem{theorem}{Theorem}[section]
\newtheorem{lemma}[theorem]{Lemma}

\newtheorem{corollary}[theorem]{Corollary}
\newtheorem{proposition}[theorem]{Proposition}

\theoremstyle{definition}
\newtheorem{definition}[theorem]{Definition}
\newtheorem{remark}[theorem]{Remark}
\newtheorem{example}[theorem]{Example}

\makeatletter 
\newcommand{\EF}[1]{\if\relax\detokenize\expandafter{\@firstofone#1{}}\relax EF\xspace\else EF#1\fi}
\makeatother
\newcommand{\EFOne}{\EF{1}\xspace}

\newcommand{\EFM}{\EF{M$_{>0}$}\xspace}
\newcommand{\EFMzero}{\EF{M}$_{\geq0}$\xspace}

\newcommand{\x}{\mathbf{x}}
\newcommand{\s}{\mathbf{s}}
\newcommand{\val}{\mathbf{v}}
\newcommand{\M}{\mathcal{M}}
\newcommand{\A}{\mathcal{A}}
\newcommand{\MNW}{\mathrm{MNW}^{\mathrm{tie}}}
\newcommand{\mbar}{\overline{m}}
\newcommand{\kbar}{\overline{k}}

\definecolor{biaoshuaiFavourate}{rgb}{0.50, 0.30, 0.07}

\title{Truthful Fair Mechanisms for Allocating Mixed Divisible and Indivisible Goods}

\author{
	Zihao Li\\
	Nanyang Technological University\\
	\texttt{zihao004@e.ntu.edu.sg}\\
    \\
    Xinhang Lu\\
	UNSW Sydney\\
	\texttt{xinhang.lu@unsw.edu.au}\\
	\and
	Shengxin Liu\\
	Harbin Institute of Technology, Shenzhen\\
	\texttt{sxliu@hit.edu.cn}\\
	\\
	Biaoshuai Tao\\
	Shanghai Jiao Tong University\\
	\texttt{bstao@sjtu.edu.cn}\\
}

\date{}

\begin{document}
\maketitle

\begin{abstract}
We study the problem of designing truthful and fair mechanisms when allocating a mixture of divisible and indivisible goods.
We first show that there does not exist an EFM (envy-free for mixed goods) and truthful mechanism in general.
This impossibility result holds even if there is only one indivisible good and one divisible good and there are only two agents.
Thus, we focus on some more restricted settings.
Under the setting where agents have binary valuations on indivisible goods and identical valuations on a single divisible good (e.g., money), we design an EFM and truthful mechanism. When agents have binary valuations over both divisible and indivisible goods, we first show there exist EFM and truthful mechanisms when there are only two agents or when there is a single divisible good.
On the other hand, we show that the mechanism maximizing Nash welfare cannot ensure EFM and truthfulness simultaneously.
\end{abstract}

\section{Introduction}

Fair allocation problem considers how to fairly allocate scarce resources among interested agents (see excellent books or surveys by, e.g., \citet{brams1995envy,RobertsonWe98,Moulin19,Suksompong21,AmanatidisAzBi22}).
This problem has gained substantial attentions in various fields including computer science, mathematics, and economics, due to the ubiquity in numerous real-world scenarios (e.g., school choices~\citep{abdulkadirouglu2005new}, course allocations~\citep{budish2012multi}, and allocating computational resources~\citep{ghodsi2011dominant}).

The literature of fair allocation problem can be categorized by the type of resources being allocated.
The first line of work studies the allocation of \emph{divisible} goods, where the famous fairness criterion \emph{envy-freeness} has been extensively studied~\citep[see, e.g.,][]{Foley67,AzizMa16}.
In an envy-free allocation, each agent weakly prefers her own bundle than any other agent's bundle.
The second group studies the allocation of \emph{indivisible} goods, in which an envy-free allocation may fail to exist.
A common practice to circumvent the issue is to consider relaxed notions such as \emph{envy-freeness up to one good (\EFOne)} in which agent $i$\rq{}s envy towards agent~$j$ could be eliminated if we (hypothetically) remove a good in agent $j$\rq{}s bundle~\citep{Lipton04onapproximately,budish2011combinatorial}.

In addition to fairness, \emph{truthfulness} is an important consideration, given that agents report their private preferences over the resources.
Roughly speaking, a mechanism is said to be truthful if each agent cannot benefit by misreporting her preference.
The truthfulness aspect of fair allocation has been addressed in a number of recent papers~\citep[e.g.,][]{bogomolnaia2004random,caragiannis2009low,kurokawa2013cut,branzei2015dictatorship,mossel2010truthful,maya2012incentive,AzizYe14,li2015truthful,amanatidis2016truthful,menon2017deterministic,bei2017cake,BeiHuSu20}.
When the resource to be allocated is a cake (i.e., a \emph{heterogeneous} divisible good), the seminal work by \citet{ChenLaPa13} designed the first truthful envy-free mechanism when each agent's valuation is piecewise-uniform.
On the other hand, very recently, \citet{BuSoTa23} showed that for piecewise-constant valuations (which is a more general type of valuation functions than piecewise-uniform functions), there does not exist a (deterministic) truthful and envy-free mechanism.
For indivisible goods setting, \citet{AmanatidisBiCh17} provided a characterization of truthful mechanisms for two agents, and further showed that truthfulness and EF1 are incompatible even for two agents and five indivisible goods.
This negative result, however, does not hold any more for some restricted cases.
With binary valuations, \citet{HalpernPrPs20} and \citet{BabaioffEzFe21} independently designed truthful and EF1 mechanisms by using maximum Nash welfare with lexicographic tie-breaking.

The aforementioned results paved the way for understanding the interplay between truthfulness and fairness for the fair allocation problem with either divisible or indivisible goods.
However, when the resources contain a mixture of both, the study of designing truthful and fair allocation mechanisms is mostly absent, which is our focus in this paper.
The only exception we know of is the work by \citet{GokoIgKa22}.
They concerned indivisible goods allocation and designed a truthful and fair mechanism that achieves envy-freeness by \emph{subsidizing} each agent with at most~$1$, the maximum marginal value of each good for each agent.

We adopt a different perspective than that of \citet{GokoIgKa22}.
To be more specific, in our setting, the divisible and indivisible goods to be allocated are both \emph{fixed} in advance~\citep{bei2021fair,BeiLiLu21,BhaskarSrVa21,LuPeAz23,NishimuraSu23}.
In a setting with mixed divisible and indivisible goods (\emph{mixed goods} for short), \citet{bei2021fair} proposed a new fairness notion called \emph{envy-freeness for mixed goods (EFM)} that generalizes both envy-freeness and EF1, and showed constructively that an EFM allocation always exists for any number of agents.
Can we go one step further by designing \emph{truthful and EFM} mechanisms when allocating mixed goods?

\subsection{Our Results}


We study the problem of designing truthful and EFM mechanisms when allocating mixed divisible and indivisible goods to agents who have additive valuations over the goods.
To the best of our knowledge, this is the first work examining the compatibility of truthfulness and EFM.
Two variants of EFM are considered in this paper, and we use \EFMzero and \EFM to distinguish them (see \Cref{sec:prelim} for their formal definitions).
Intuitively speaking, \EFMzero requires that the envy-free criterion is imposed even if the envied bundle contains a \emph{positive amount} of divisible goods and the \EFOne criterion is used otherwise.
Slightly differently, \EFM only requires to impose the envy-free criterion if the envied bundle contains divisible goods with \emph{positive value}.
It can be verified that \EFMzero implies \EFM.
While in the following we mostly present our results regarding to \EFM, some of the results can be extended to the case of \EFMzero.

We start by giving in \Cref{sec:general-impossibility} a strong impossibility result showing that truthfulness and \EFM are incompatible even if there are only two agents and the goods to be allocated consist of only a single indivisible good and a single divisible good.
Since we can normalize the valuations so that agents' valuations on the indivisible good are~$0$ or~$1$, as a corollary to the impossibility result, truthfulness and \EFM are incompatible for two agents with \emph{binary} valuations on indivisible goods.
Truthfulness and \EFM, however, are compatible if we further restrict the expressiveness of agents' valuations on the divisible goods.

First, in \Cref{sec:binary-IND+money}, we design a truthful and \EFM mechanism when agents have binary valuations over indivisible goods and an \emph{identical} valuation (not necessarily binary) over a single divisible good.
In addition, the allocations output by our mechanism satisfy some nice efficiency properties including leximin and Maximum Nash Welfare (MNW).
Next, in \Cref{sec:binary-IND-and-DIV}, we consider the case where agents have binary valuations over all the goods.
Specifically, we design truthful and \EFM mechanisms when (i) there are \emph{two} agents (and an arbitrary number of goods), or (ii) the mixed goods consist of an arbitrary number of indivisible goods and a single divisible good.
Technically speaking, in general, our mechanisms first make use of the truthful and \EFOne mechanism of \citet{HalpernPrPs20} to allocate the (binary) indivisible goods, and next design different methods to allocate the divisible good(s) in the three different scenarios described above.

\section{Preliminaries}
\label{sec:prelim}

Let $[s] \coloneqq \{1, \dots, s\}$.
We use $N = [n]$ to denote the set of~$n$ agents.
The set of the goods is denoted by $(G, D)$, where $G = \{g_1, \dots, g_{m}\}$ is the set of~$m$ \emph{indivisible goods} and $D = \{d_1, \dots, d_{\mbar}\}$ is the set of~$\mbar$  \emph{divisible goods}.
Each divisible good is \emph{homogeneous}, meaning that an agent's value on each divisible good only depends on the \emph{fraction} of this divisible good allocated to her.
(We will formally define the valuations of the agents later.)
Denote by $\A = (A_1, \dots, A_n)$ an \emph{allocation}, where we assign bundle~$A_i$ to agent~$i$.
Each~$A_i$ is composed by a pair $(G_i, \x_i)$, where $G_i$ is a subset of the indivisible goods allocated to agent~$i$ and $\x_i = (x_{i1}, \dots, x_{i\mbar})$ specifies how divisible goods are allocated to agent~$i$---specifically, $x_{i\kbar}$ denotes the fraction of (homogeneous) divisible good $d_{\kbar}$ allocated to agent~$i$.
Naturally, an allocation $\A = ((G_1, \x_1), \dots, (G_n, \x_n))$ must satisfy that $(G_1, \dots, G_n)$ is a partition of~$G$ and that $\sum_{i = 1}^n x_{i\kbar} = 1$ for each $\kbar = 1, \dots, \mbar$ (in particular, we have assumed each divisible good has $1$ unit of amount).

We assume that each agent~$i \in N$ has an additive valuation function~$v_i$ and call $(v_1, \dots, v_n)$ a \emph{valuation profile}.
That is, agent~$i$\rq{}s value on a bundle $(G_j, \x_j)$ is given by
\[
v_i(G_j, \x_j) = \sum_{g_k \in G_j} v_i(g_k) + \sum_{\kbar = 1}^{\mbar} x_{j \kbar} \cdot v_i(d_{\kbar}),
\]
where $v_i(g_k)$ is agent~$i$\rq{}s value on the indivisible good~$g_k$ and $v_i(d_{\kbar})$ is agent~$i$\rq{}s value on the divisible good~$d_{\kbar}$.
We slightly abuse the notation by letting $v_i(G_j) = \sum_{g_k \in G_j} v_i(g_k)$ and $v_i(\x_j) = \sum_{\kbar = 1}^{\mbar} x_{j \kbar} \cdot v_i(d_{\kbar})$.
We say that agents\rq{} valuations are \emph{binary} if $v_i(g_k) \in \{0, 1\}$ and $v_i(d_{\kbar}) \in \{0, 1\}$ for every $i$, $k$ and $\kbar$.
Correspondingly, we say that agents have \emph{binary valuations on indivisible goods} if $v_i(g_k) \in \{0, 1\}$ for every~$i$ and~$k$ and agents have \emph{binary valuations on divisible goods} if $v_i(d_{\kbar}) \in \{0, 1\}$ for every~$i$ and~$\kbar$.

An allocation is \emph{envy-free} if each agent believes (according to her own valuation) her own allocated bundle is weakly more valuable than that of every other agent's.
In our case with both divisible and indivisible goods, this means $v_i(G_i, \x_i) \geq v_i(G_j, \x_j)$ holds for every pair of agents~$i$ and~$j$.
An envy-free allocation may not exist even if there are only indivisible goods (i.e., $D = \emptyset$).
For allocating only indivisible goods, a commonly adopted relaxation of envy-freeness is \emph{envy-freeness up to one item} (\EFOne), and it is well-known that an \EFOne allocation always exists~\citep{Lipton04onapproximately,budish2011combinatorial}.

\begin{definition}[\EFOne]\label{def:EF1}
For $D = \emptyset$, given a valuation profile $(v_1, \dots, v_n)$, an allocation $(G_1, \dots, G_n)$ is \EFOne if for any pair of $i, j \in N$, there exists a good $g \in G_j$ such that $v_i(G_i)\geq v_i(G_j \setminus\{g\})$.
\end{definition}

With mixed goods, we adopt the fairness notion called \emph{envy-freeness for mixed goods (EFM)}~\citep{bei2021fair}.
There are two variants of EFM, and we use \EFMzero and \EFM to distinguish them.
As we will see shortly, \EFMzero implies \EFM.
We start by presenting the intuition of \EFM: For any pair of agents $i, j \in N$, agent~$i$ should not envy agent~$j$, as stated in Point~2 of \Cref{def:EFM} below, with the exception that the divisible part allocated to agent~$j$ is worthless to agent~$i$, in which case the \EFOne condition holds for the indivisible part (see Point~1 of \Cref{def:EFM}).

\begin{definition}[\EFM]\label{def:EFM}
Given a valuation profile $(v_1,\dots,v_n)$, an allocation $((G_1,\x_1),\ldots,(G_n,\x_n))$ is \EFM if the followings hold for any pair of $i, j \in N$.
\begin{enumerate}
\item If $v_i(\x_j) = 0$ and $G_j\neq\emptyset$, then there exists a good $g \in G_j$ such that $v_i(G_i,\x_i)\geq v_i(G_j\setminus\{g\},\x_j)$.
\item Otherwise, $v_i(G_i,\x_i)\geq v_i(G_j,\x_j)$.
\end{enumerate}
\end{definition}

\EFMzero, on the other hand, imposes envy-free condition as long as agent~$j$'s bundle has any positive amount of divisible goods, even if agent~$i$ values the divisible part at~$0$:
\begin{definition}[\EFMzero]\label{def:EFMzero}
Given a valuation profile $(v_1,\dots,v_n)$, an allocation $((G_1,\x_1),\ldots,(G_n,\x_n))$ is \EFMzero if the followings hold for any pair of $i, j \in N$.
\begin{enumerate}
\item If $\x_j=\mathbf{0}$ and $G_j\neq\emptyset$, then there exists a good $g \in G_j$ such that $v_i(G_i,\x_i)\geq v_i(G_j\setminus\{g\},\x_j)$.
\item Otherwise, $v_i(G_i,\x_i)\geq v_i(G_j,\x_j)$.
\end{enumerate}
\end{definition}

It is easy to see that \EFMzero implies \EFM.
It also directly follows from the above definitions that when there are only indivisible goods (i.e., $D=\emptyset$), \EFM and \EFMzero reduce to \EFOne; when there are only divisible goods (i.e., $G=\emptyset$), \EFM and \EFMzero reduce to envy-freeness.

A \emph{mechanism} is a function $\M$ that maps the set of $n$ valuation functions $(v_1, \dots, v_n)$ to an allocation $((G_1, \x_1), \dots, (G_n, \x_n))$.
We only consider deterministic mechanisms in this paper.
In the game-theoretical setting, each agent $i$ submits a valuation $v_i\rq{}$ to $\M$ which may or may not be her true valuation $v_i$.
A mechanism $\M$ is \EFM/\EFMzero if, upon receiving every input $(v_1\rq{}, \dots, v_n\rq{})$, it outputs an allocation that is \EFM/\EFMzero with respect to $(v_1\rq{}, \dots, v_n\rq{})$.
A mechanism $\M$ is \emph{truthful} if it is each agent $i$\rq{}s dominant strategy to truthfully report her valuation $v_i$.
Formally, let $v_i$ be the true valuation function of an arbitrary agent $i$ and $v_i\rq{}$ be an arbitrary valuation function, for any $n-1$ valuation functions $v_1, \dots, v_{i-1}, v_{i+1}, \dots, v_n$ of the remaining $n-1$ agents, we have
$$v_i(G_i, \x_i) \geq v_i(G_i\rq{}, \x_i\rq{}),$$
where $(G_i, \x_i)$ is the bundle allocated to agent $i$ by $\M$ when receiving the input $(v_1, \dots, v_{i-1}, v_i, v_{i+1}, \dots, v_n)$ and $(G_i\rq{}, \x_i\rq{})$ is the bundle allocated to agent $i$ by $\M$ when receiving the input $(v_1, \dots, v_{i-1}, v_i\rq{}, v_{i+1}, \dots, v_n)$.

We make the standard \emph{free-disposal} assumption \citep[see, e.g.,][]{ChenLaPa13,bei2017cake,HalpernPrPs20}, which assumes that a good $g_k$ (resp., $d_{\kbar}$) is discarded by the mechanism if $v_i(g_k)=0$ (resp., $v_i(d_{\kbar})=0$) for all $i \in N$.
Without this assumption, we may run into uninteresting technicality which will be further discussed in Section~\ref{secmnw}.

\subsection{Maximum Nash Welfare and Leximin}
\label{secmnw}

We now proceed to review the concepts and some properties of \emph{Maximum Nash Welfare} (MNW) allocations and \emph{leximin} allocations, which will be useful in our paper.

The definitions of MNW and leximin allocations apply to general fair division settings.
For each $i \in [n]$, if we are only allocating indivisible goods, then $A_i = G_i$; if we are only allocating divisible goods, then $A_i = \x_i$; for mixed indivisible and divisible goods, $A_i = (G_i, \x_i)$.

\begin{definition}[MNW]
\label{def:MNW}
Given a valuation profile $(v_1,\dots,v_n)$, an allocation $(A_1,\dots,A_n)$ is a \emph{Maximum Nash Welfare (MNW)} allocation if it first maximizes the number of the agents receiving positive values, i.e., $|\{i \in [n] : v_i(A_i) > 0\}|$, and, subject to this, maximizes the product of the positive utilities, i.e., $\prod_{i : v_i(A_i) > 0} v_i(A_i)$.
\end{definition}

When allocating only divisible goods, an MNW allocation is always envy-free~\citep{varian1974equity}.
When allocating only indivisible goods, an MNW allocation is always \EFOne~\citep{CaragiannisKuMo19}.
However, with mixed goods, an MNW allocation may not be \EFM~\citep{bei2021fair}.

\begin{definition}[Leximin]\label{def:leximin}
Given two vectors $\s_1, \s_2 \in \mathbb{R}^n$, let $\s_1\rq{}$ and $\s_2\rq{}$ be the vectors obtained by sorting $\s_1$ and $\s_2$ in ascending order respectively.
We say that $\s_1$ \emph{leximin-dominates} $\s_2$ if there exists $i \in [n]$ such that $\s_1\rq{}$ and $\s_2\rq{}$ are identical for the first $i-1$ entries and the $i$-th entry of $\s_1\rq{}$ is greater than the $i$-th entry of $\s_2\rq{}$.
Given a valuation profile $(v_1,\dots,v_n)$, an allocation $(A_1,\ldots,A_n)$ is a \emph{leximin} allocation if $(v_1(A_1),\dots,v_n(A_n))$ is maximum in the total order induced by leximin-domination among all allocations.
\end{definition}

In other words, a leximin allocation maximizes the minimum among the agents' utilities; among all such allocations, it considers those maximizing the second smallest utility, and so on.

When allocating indivisible goods \emph{with binary valuations}, \cite{AzizRe20} and \cite{HalpernPrPs20} showed the equivalence of \emph{Maximum Nash Welfare} (MNW) allocations and \emph{leximin} allocations.
In addition, \cite{HalpernPrPs20} showed that the mechanism that outputs the MNW/leximin allocation with a consistent lexicographic tie-breaking rule is truthful.
For the purpose of our paper, we state \citeauthor{HalpernPrPs20}\rq{}s mechanism (referred to as $\MNW$) below.

\begin{theorem}[\citet{HalpernPrPs20}, $\MNW$]\label{thm:MNWtruthful}
For allocating only indivisible goods with binary valuations, there exists a truthful mechanism that always outputs an allocation that is both MNW and leximin.
\end{theorem}

\paragraph{The free-disposal assumption.}
In the following example, we show that $\MNW$ may not be truthful without the free-disposal assumption.
    \begin{example}
        We need to allocate five indivisible goods to three agents. The specific valuation profile is listed in the table.
        \begin{table}[H]
            \centering
            \begin{tabular}{@{}*{7}{c}@{}}
            \toprule
            & $g_1$ & $g_2$ & $g_3$ & $g_4$ & $g_5$\\
            \midrule
            $v_1$ & $1$ & $1$ & $1$ & $1$ & $1$\\
            $v_2$ & $1$ & $1$ & $1$ & $0$ & $1$\\
            $v_3$ & $1$ & $1$ & $1$ & $0$ & $1$\\
            \bottomrule
            \end{tabular}
            \end{table}
        In this example, if the tie-breaking rule is to prioritize agent $1$, the allocation is to allocate $g_4$ and one other item to agent $1$, and to allocate the remaining three items to agents $2$ and $3$.
        Without the free-disposal assumption, if we allocate the worthless item to agent $1$, this will lead to a possible misreport.
        When agent $1$ misreports her valuation towards $g_4$ (from $1$ to $0$), under $\MNW$, agent $1$ can still get two items among $\{g_1,g_2,g_3,g_5\}$, while the additional $g_4$ can bring her additional value.
    \end{example}
    In this example, we can see that an arbitrary allocation of globally worthless items may make $\MNW$ no longer truthful. This is just an uninteresting corner case, which is also a reason why we make the free-disposal assumption.

\section{General Impossibility Results}
\label{sec:general-impossibility}

Unfortunately, truthfulness and \EFM are incompatible in general.
In this section, we prove that there does not exist a truthful and \EFM mechanism even under the very restricted settings where there are only one indivisible good and one divisible good and there are only two agents.

Our impossibility result for mixed divisible and indivisible goods holds for a minimum number of agents (which is $2$) and a minimum number of goods (which is $2$).
This is in contrast with \cite{AmanatidisBiCh17}'s  $2$-agent-$5$-good impossibility result for only indivisible goods.

\begin{theorem}\label{thm:impossibilityGeneral}
There does not exist a truthful and \EFM (and thus \EFMzero) mechanism even when there are only two agents and the set of goods consists of one indivisible good and one divisible good.
\end{theorem}
\begin{proof}
It suffices to prove the statement for \EFM.
Consider the instance with two agents $\{1, 2\}$, one indivisible good and one divisible good.
Both agents have value~$1$ on the indivisible good, and have value~$a$ and~$b$ on the divisible good, where $b>a>1$.

Firstly, we prove that all the \EFM allocations must allocate the indivisible good to agent~$1$.
Suppose for the sake of contradiction that the indivisible good is allocated to agent~$2$.
To guarantee \EFM, agent~$2$ must get at least a fraction $\frac{b-1}{2b}$ from the divisible good in order to not envy agent~$1$, and agent~$2$ must get at most a fraction of $\frac{a-1}{2a}$ from the divisible good in order to avoid that agent~$1$ envies agent~$2$. This is impossible as $\frac{a-1}{2a}<\frac{b-1}{2b}$.

Therefore, the possible \EFM allocations can be described as follows. Agent~$1$ receives the indivisible good and a fraction~$x$ of the divisible good, and agent~$2$ receives a fraction~$1-x$ of the divisible good.
For the reasons similar as above, we must have $\frac{a-1}{2a}\leq x\leq\frac{b-1}{2b}$ to guarantee \EFM.
If the mechanism outputs an allocation with $x=\frac{a-1}{2a}$, agent $1$ can misreport her valuation by increasing~$a$, which increases agent~$1$\rq{}s received value as~$x$ increases.
If the mechanism outputs an allocation with $x>\frac{a-1}{2a}$, agent~$2$ can misreport her valuation by decreasing~$b$ so that $\frac{a-1}{2a}<\frac{b-1}{2b}<x$.
In this case, $x$ will decrease and agent~$2$ receives more value.
\end{proof}

When there is only one indivisible good and one divisible good, we can assume without loss of generality that agents' valuations on the indivisible good are binary (as we did in the proof of Theorem~\ref{thm:impossibilityGeneral}) or agents' valuations on the divisible good are binary.
This is because we can normalize the valuations of the agents.
Thus, Theorem~\ref{thm:impossibilityGeneral} straightforwardly implies the following corollary.

\begin{corollary}\label{cor:impossibilityGeneral}
There does not exist a truthful and \EFM (and thus \EFMzero) mechanism even when there are two agents and agents have binary valuations on either the indivisible goods or the divisible goods.
\end{corollary}

As a remark, although we focus on deterministic mechanisms in this paper, Theorem~\ref{thm:impossibilityGeneral} and Corollary~\ref{cor:impossibilityGeneral} continue to hold for randomized mechanisms (that are universally \EFM and truthful in expectation).\footnote{Loosely speaking, \emph{universally \EFM} mechanisms randomize over deterministic \EFM mechanisms. A randomized mechanism is said to be \emph{truthful in expectation} if misreporting a valuation function cannot increase the expected utility of an agent.}
The proof is almost the same: We must allocate the indivisible good to agent~$1$ in order to guarantee universally \EFM, and $x$ in the proof becomes the \emph{expected} fraction of the divisible good allocated to agent~$1$.

The strong impossibility results suggest that truthfulness and \EFM may only be compatible in more restrictive settings.
We confirm our intuitions in the affirmative in the following sections by considering (1) the setting with a single divisible good of identical value to all agents and multiple indivisible goods on which agents have binary valuations (\Cref{sec:binary-IND+money}), and (2) the setting where agents' valuations are binary for \emph{both} indivisible and divisible goods (\Cref{sec:binary-IND-and-DIV}).

\section{Binary Valuations on Indivisible Goods and Identical Valuation on Single Divisible Good}
\label{sec:binary-IND+money}

In this section, we consider the setting where agents\rq{} valuations on indivisible goods are binary and there is one divisible good on which agents have an identical valuation.
This describes the natural scenario where we are allocating a set of indivisible goods and some amount of \emph{money}.
Here, the divisible good is just a sum of money, and an agent\rq{}s value on each indivisible good is described by the amount of money the agent is willing to pay for the good.
We will see that there exists a truthful and \EFMzero (and thus \EFM) mechanism under this setting.

For the ease of notation, we will use $x_1,\ldots,x_n$ to denote the fractions of the (unique) divisible good allocated to the $n$ agents.
Thus, each agent's allocated share is denoted by $(G_i,x_i)$.
We will use $u$ to denote each agent\rq{}s value on the (unique) divisible good.

\begin{algorithm}[t]
    \floatname{algorithm}{Mechanism}
    \caption{A truthful \EFMzero mechanism for binary valuations on indivisible good and identical valuation on a single divisible good}
    \label{alg:oneidentical}
    \begin{algorithmic}[1]
        \STATE use $\MNW$ (Theorem~\ref{thm:MNWtruthful}) to compute an MNW/leximin allocation $(G_1,\ldots,G_n)$ of $G$
        \STATE initialize $x_i\leftarrow 0$ for each agent $i$
        \STATE \textbf{while} $y:=1-\sum_{i=1}^nx_i>0$:
        \STATE \hspace{0.1cm} let $T_1=\arg\min_{i\in[n]}\{v_i(G_i,x_i)\}$
        \STATE \hspace{0.1cm} $\Delta\leftarrow \infty$
        \STATE \hspace{0.1cm} \textbf{if }$T_1\neq [n]$ \textbf{then}
        \STATE \hspace{0.4cm} let $T_2=\arg\min_{i\in[n]\setminus T_1}\{v_i(G_i,x_i)\}$
        \STATE \hspace{0.4cm} $\Delta \leftarrow v_j(G_j,x_j)-v_i(G_i,x_i)$ for $i\in T_1$, $j\in T_2$
        \STATE \hspace{0.1cm} for each $i\in T_1$, $x_i\leftarrow x_i+\min\{\frac\Delta{u},\frac{y}{|T_1|}\}$
        \STATE \textbf{return} allocation $((G_1,x_1),\ldots,(G_n,x_n))$
    \end{algorithmic}
\end{algorithm}

Our mechanism is shown in Mechanism~\ref{alg:oneidentical}.
The mechanism first allocates the indivisible goods by using $\MNW$ (Theorem~\ref{thm:MNWtruthful}) and then iteratively allocates the divisible good.
In each iteration, the mechanism identifies a set $T_1$ of agents who receive minimum values in the current partial allocation.
The mechanism then attempts to compensate agents in $T_1$ with some fraction of the unallocated divisible good.
This is done by a ``water-filling'' process:
we allocate the divisible good to the agents in $T_1$ at an equal rate, until the divisible good is fully allocated, in which case the mechanism terminates, or until the value received by each agent in $T_1$ reaches the second minimum value received among all agents (i.e., the value received by an agent in $T_2$), in which case the mechanism proceeds to the next iteration.

\begin{theorem}
Mechanism~\ref{alg:oneidentical} is \EFMzero (and thus \EFM) and truthful. Moreover, it always outputs allocations that are both leximin and MNW.
\end{theorem}

The following three subsections aim to prove this theorem.

\subsection{Mechanism~\ref{alg:oneidentical} is \EFMzero}
We first present some simple yet important observations.

\begin{proposition}\label{prop:nonwasteful}
    Let $((G_1,x_1),\ldots,(G_n,x_n))$ be the allocation output by Mechanism~\ref{alg:oneidentical}. The followings hold.
    \begin{enumerate}
        \item For any agent $i$, $v_i(G_i,x_i)=|G_i|+u\cdot x_i$.
        \item For any agents $i$ and $j$, $v_i(G_i,x_i)\geq v_j(G_i,x_i)$.
    \end{enumerate}
\end{proposition}
\begin{proof}
    Since agents have binary valuations on indivisible goods and the identical valuation $u$ on the unique divisible good, we have $v_j(G_i,x_i)\leq |G_i|+u\cdot x_i$ for any agents $i$ and $j$.
    Therefore, Point~1 implies Point~2, and it remains to show Point~1.
    To show Point~1, notice that $v_i(G_i,x_i)\neq |G_i|+u\cdot x_i$ is only possible when there exists $g\in G_i$ such that $v_i(g)=0$.
    Suppose this is the case for the sake of contradiction.
    If there exists another agent~$j$ with $v_j(g)=1$, then moving~$g$ from~$i$'s bundle to~$j$'s increases the Nash welfare, which contradicts to that $(G_1,\ldots,G_n)$ is an MNW allocation.
    If the good~$g$ is worth~$0$ to all agents, then~$g$ is discarded by the free-disposal assumption, which contradicts to $g\in G_1$.
\end{proof}

\begin{lemma}
Mechanism~\ref{alg:oneidentical} always outputs \EFMzero (and thus \EFM) allocations.
\end{lemma}
\begin{proof}
    We will prove by induction that the partial allocation is \EFMzero after each while-loop iteration of the mechanism.
    For the base step, the MNW allocation $(G_1,\ldots,G_n)$ for the indivisible goods is \EFOne by \citet{CaragiannisKuMo19}, and \EFMzero is satisfied since the allocation of the divisible good has not been started.

    For the inductive step, suppose the partial allocation $((G_1,x_1),\ldots,(G_n,x_n))$ is \EFMzero before a while-loop iteration.
    After one while-loop iteration, each agent in $T_1$ receives an extra fraction $z:=\min\{\frac\Delta{u},\frac{y}{|T_1|}\}$ of the divisible good.
    It suffices to show that $j$ does not envy $i$ for any agent $i\in T_1$ and any other agent $j$.
    We discuss two cases: $j\in T_1$ and $j\notin T_1$.

    If $j\in T_1$, then we have $v_j(G_j,x_j)=v_i(G_i,x_i)\geq v_j(G_i,x_i)$ before the execution of this while-loop iteration, where the first equality is due to our definition of $T_1$ in Mechanism~\ref{alg:oneidentical} and the second inequality is due to 2 of Proposition~\ref{prop:nonwasteful}.
    This implies that agent $j$ does not envy agent $i$ before the while-loop iteration.
    Agent $j$ will not envy agent $i$ after the while-loop iteration, as both agents receive the same amount $z$ of the divisible good, which is worth the same value $u\cdot z$ to both agents.

    If $j\notin T_1$, then we have $v_j(G_j,x_j)\geq \Delta+v_i(G_i,x_i)\geq \Delta+v_j(G_i,x_i)$ before the execution of this while-loop iteration, where, again, the first equality is due to our definition of $T_1$ in Mechanism~\ref{alg:oneidentical} and the second inequality is due to Point 2 of Proposition~\ref{prop:nonwasteful}.
    This implies that agent $j$ will not envy agent $i$ after the while-loop iteration if the portion of the divisible goods allocated to agent $i$ is worth at most $\Delta$.
    This is true as $u \cdot z\leq u\cdot\frac\Delta{u}=\Delta$.
\end{proof}

\subsection{Mechanism~\ref{alg:oneidentical} is truthful}
We first define a type of mechanisms called \emph{water-filling mechanisms}.
A water-filling mechanism starts from an allocation $(G_1,\ldots,G_n)$ of the indivisible goods and then proceeds to allocate the unique divisible good by the ``water-filling process'' defined by the while-loop in Mechanism~\ref{alg:oneidentical}.
Our Mechanism~\ref{alg:oneidentical} is a \emph{particular} water-filling mechanism by specifying that the allocation of the indivisible goods $(G_1,\ldots,G_n)$ is output by $\MNW$.

\begin{definition}
    Given a valuation profile $(v_1,\ldots,v_n)$,  an allocation $((G_1,x_1),\ldots,(G_n,x_n))$ \emph{satisfies the water-filling property} if
    \begin{enumerate}
        \item for any two agents $i$ and $j$ with $x_i>0$ and $x_j>0$, we have $v_i(G_i,x_i)=v_j(G_j,x_j)$, and
        \item for any two agents $i$ and $j$ with $x_i=0$ and $x_j>0$, we have $v_i(G_i,x_i)\geq v_j(G_j,x_j)$.
    \end{enumerate}
\end{definition}

It is straightforward to check that the allocation output by any water-filling mechanism satisfies the water-filling property.
Given a valuation profile $\val=(v_1,\ldots,v_n)$ and an allocation $A=((G_1,x_1),\ldots,(G_n,x_n))$ satisfying the water-filling property, we define the \emph{potential} $\phi(A,\val)$ by the ``height of the water level'':
$\phi(A,\val)=v_i(G_i,x_i)$,
where $i$ is an arbitrary agent with $x_i>0$.
When $x_i=0$ for all $i\in[n]$, $\phi(A,\val)=\min_{i\in[n]}v_i(G_i)$.

We will first prove a proposition, Proposition~\ref{prop:waterfillingoptimal}, which follows from that an MNW allocation of indivisible goods with binary valuations is always \emph{Lorenz dominating}~\citep{BabaioffEzFe21}.
Before stating the proposition, we will first state Babaioff \emph{et al.}'s result.

\begin{definition}\label{def:water-filling}
    Given a valuation profile $(v_1,\ldots,v_n)$, an allocation $(A_1,\ldots,A_n)$ is \emph{Lorenz dominating} if, for any $k\in [n]$,   the sum of the $k$ smallest values of $v_1(A_1),\ldots,v_n(A_n)$ is weakly larger than the sum of the $k$ smallest values of $v_1(A_1'),\ldots,v_n(A_n')$ for any allocation $(A_1',\ldots,A_n')$.
\end{definition}

A Lorenz dominating allocation may not exist. However, if it does, it is easy to see that a Lorenz dominating allocation is always leximin.
\cite{BabaioffEzFe21} (implicitly) proved the following theorem.

\begin{proposition}[\citep{BabaioffEzFe21}]\label{thm:Lorenz}
    Consider the allocation of indivisible goods (i.e., $D=\emptyset$). For any binary valuation profile on indivisible goods, an MNW/leximin allocation is Lorenz dominating.
\end{proposition}

In particular, Halpern \emph{et al.}'s mechanism $\MNW$ always outputs Lorenz dominating allocations.
We then get the following proposition from the Lorenz domination.

\begin{proposition}\label{prop:waterfillingoptimal}
    Fix a valuation profile $(v_1,\ldots,v_n)$. Mechanism~\ref{alg:oneidentical} outputs an allocation $A$ that maximizes the potential $\phi(A,\val)$ among all allocation satisfying the water-filling property.
\end{proposition}
\iftrue
\begin{proof}
    Let $A=((G_1,x_1),\ldots,(G_n,x_n))$ be the allocation output by Mechanism 1 and $A'=((G_1',x_1'),\ldots,(G_n',x_n'))$ be an arbitrary allocation satisfying the water-filling property.
    Let $a_1,\ldots,a_n$ be an ordering of the $n$ agents such that $v_{a_1}(G_{a_1})\leq \cdots\leq v_{a_n}(G_{a_n})$, and $b_1,\ldots,b_n$ be an ordering of the $n$ agents such that $v_{b_1}(G_{b_1}')\leq \cdots\leq v_{b_n}(G_{b_n}')$.
    Since $(G_1,\ldots,G_n)$ is Lorenz dominating (Proposition 4.6), for any $k\in[n]$, we have $\sum_{i=1}^kv_{a_i}(G_{a_i})\geq\sum_{i=1}^kv_{b_i}(G_{b_i}')$.

    Let $h=\phi(A,\mathbf{v})$ and $h'=\phi(A',\mathbf{v})$.
    We also let $\ell$ be the number of agents with $x_i>0$ and $\ell'$ be the number of agents with $x_i'>0$.

    By the water-filling property, for each $i=1,\ldots,\ell$, we have $v_{a_i}(G_{a_i})+u\cdot x_{a_i}=h$, and $\ell\cdot h=u+\sum_{i=1}^\ell v_{a_i}(G_{a_i})$.
    Similarly, for each $i=1,\ldots,\ell'$, we have $v_{b_i}(G_{b_i}')+u\cdot x_{b_i}'=h'$, and $\ell'\cdot h'=u+\sum_{i=1}^{\ell'} v_{b_i}(G_{b_i}')$.

    Suppose for the sake of contradiction that $h'>h$. We discuss two cases: $\ell'\geq\ell$ and $\ell'<\ell$.

    If $\ell'\geq\ell$, we have $\ell\cdot h'=u\cdot\sum_{i=1}^{\ell}x_{b_i}'+\sum_{i=1}^{\ell} v_{b_i}(G_{b_i}')\leq u+\sum_{i=1}^{\ell} v_{b_i}(G_{b_i}')$.
    On the other hand, $\ell\cdot h'>\ell\cdot h=u+\sum_{i=1}^\ell v_{a_i}(G_{a_i})$.
    Combining two inequalities yields $\sum_{i=1}^{\ell} v_{b_i}(G_{b_i}')>\sum_{i=1}^\ell v_{a_i}(G_{a_i})$, which contradicts to Lorenz domination.

    If $\ell'<\ell$, for each $i=\ell'+1,\ldots,\ell$, we have $v_{b_i}(G_{b_i}')\geq h'$.
    Therefore, $\ell\cdot h'\leq\ell'\cdot h'+\sum_{i=\ell'+1}^{\ell} v_{b_i}(G_{b_i}')= u+\sum_{i=1}^{\ell} v_{b_i}(G_{b_i}')$.
    On the other hand, $\ell\cdot h'>\ell\cdot h=u+\sum_{i=1}^\ell v_{a_i}(G_{a_i})$.
    Combining two inequalities yields $\sum_{i=1}^{\ell} v_{b_i}(G_{b_i}')>\sum_{i=1}^\ell v_{a_i}(G_{a_i})$, which again contradicts to Lorenz domination.
\end{proof}
\fi

\begin{lemma}
    Mechanism~\ref{alg:oneidentical} is truthful.
\end{lemma}
\begin{proof}
    Let $\M$ be Mechanism~\ref{alg:oneidentical}.
    Consider a valuation profile $\val=(v_1,\ldots,v_n)$ and suppose agent $1$ misreports her valuation to $v_1'$.
    Let $\val'$ be the valuation profile $(v_1',v_2,\ldots,v_n)$.
    Let $A=((G_1,x_1),\ldots,(G_n,x_n))$ be the allocation output by $\M$ when agent $1$ reports truthfully, and let $A'=((G_1',x_1'),\ldots,(G_n',x_n'))$ be the allocation output by $\M$ when agent $1$ reports $v_1'$.
    By the truthfulness of $\MNW$, we have $v_1(G_1)\geq v_1(G_1')$.
    Let $A^\dag=((G_1',x_1^\dag),\ldots,(G_n',x_n^\dag))$ be the allocation obtained by applying the water-filling process (the while-loop in Mechanism~\ref{alg:oneidentical})  to the start-up allocation $(G_1',\ldots,G_n')$ with the true valuation profile $\val$ considered.
    Both $A$ and $A^\dag$ satisfy the water-filling property with respect to $\val$, and $A'$ satsifies the water-filling property with respect to $\val'$.

    If $x_1'=0$, we have $v_1(G_1,x_1)\geq v_1(G_1)\geq v_1(G_1')=v_1(G_1',x_1')$, and the truthfulness of $\M$ holds trivially.
    Thus, we assume $x_1'>0$ from now on.

    By Proposition~\ref{prop:waterfillingoptimal}, $\phi(A,\val)\geq\phi(A^\dag,\val)$.
    We have $v_1(G_1,x_1)\geq \phi(A,\val)$ according to the definition of the potential function $\phi$.
    To conclude that $\M$ is truthful, it suffices to show that $\phi(A^\dag,\val)\geq v_1(G_1',x_1')$.

    Since we have assumed $x_1'>0$, $v_1'(G_1',x_1')=\phi(A',\val')$.
    Let $\delta = v_1'(G_1',x_1')-v_1(G_1',x_1')$.
    It then remains to show that $\phi(A^\dag,\val)+\delta\geq \phi(A',\val')$.

    By Point 1 of Proposition~\ref{prop:nonwasteful}, we have $v_1'(g)=1$ for each $g\in G_1'$.
    Moreover, it is clear that $\delta$ equals to the number of goods $g$ in $G_1'$ with $v_1(g)=0$ and $v_1'(g)=1$.
    In particular, $\delta\geq 0$.
    Now, consider the two water-filling processes corresponding to $\phi(A^\dag,\val)$ and $\phi(A',\val')$.
    When the ``height of the water level'' reaches $\phi(A^\dag,\val)$, the first process terminates, while an additional amount $\delta$ of water is yet to be split among one or more agents in the second process.
    Thus, the ``height of the water level'' for the second process can be further increased by at most $\delta$.
    Therefore, $\phi(A^\dag,\val)+\delta\geq \phi(A',\val')$.
\end{proof}

\subsection{Leximin and MNW}

We show that the allocation output by Mechanism~\ref{alg:oneidentical} is Lorenz dominating.
\begin{proposition}
    Given a valuation profile $\val=(v_1,\ldots,v_n)$, the allocation $A=((G_1,x_1),\ldots,(G_n,x_n))$ output by Mechanism~\ref{alg:oneidentical} is Lorenz dominating.
\end{proposition}
\begin{proof}
    Suppose there is an allocation $A'=((G_1',x_1'),\ldots,(G_n',x_n'))$ such that $A$ does not Lorenz dominate $A'$.
    Let $a_1,\ldots,a_n$ be an ordering of the $n$ agents such that $v_{a_1}(G_{a_1},x_{a_1})\leq \cdots\leq v_{a_n}(G_{a_n},x_{a_n})$, and $b_1,\ldots,b_n$ be an ordering of the $n$ agents such that $v_{b_1}(G_{b_1}',x_{b_1}')\leq \cdots\leq v_{b_n}(G_{b_n}',x_{b_n}')$.
    Let $k$ be the smallest index such that
    \begin{equation}\label{eqn:lorenz}
        \sum_{i=1}^kv_{a_i}(G_{a_i},x_{a_i})<\sum_{i=1}^kv_{b_i}(G_{b_i}',x_{b_i}').
    \end{equation}

    We can first assume without loss of generality that $A'$ satisfies the water-filling property.
    If not, we can adjust $A'$ by applying the water-filling process to the start up allocation $(G_1',\ldots,G_n')$.
    It is easy to see that the adjusted allocation Lorenz dominates the original allocation.

    By Proposition~\ref{prop:waterfillingoptimal}, we have $\phi(A,\val)\geq\phi(A',\val)$.
    Let $\ell$ be the number of agents with $x_i>0$ and $\ell'$ be the number of agnets with $x_i'>0$.
    We must have $k>\ell'$.
    Otherwise, $\sum_{i=1}^kv_{b_i}(G_{b_i}',x_{b_i}')=k\cdot\phi(A',\val)\leq k\cdot\phi(A,\val)$.
    Since $\phi(A,\val)$ is a lower bound to each $v_i(G_i,x_i)$, Equation (\ref{eqn:lorenz}) cannot be true.
    Since $k>\ell'$, we have
    \begin{equation}\label{eqn:lorenzb}
        \sum_{i=1}^kv_{b_i}(G_{b_i}',x_{b_i}')=u+\sum_{i=1}^kv_{b_i}(G_{b_i}').
    \end{equation}

    Now we consider two cases: $\ell\leq k$ and $\ell>k$.
    If $\ell\leq k$,  we must have $\sum_{i=1}^kv_{a_i}(G_{a_i},x_{a_i})=u+\sum_{i=1}^kv_{a_i}(G_{a_i})$.
    This, together with Equations (\ref{eqn:lorenz}) and (\ref{eqn:lorenzb}), implies $\sum_{i=1}^kv_{a_i}(G_{a_i})<\sum_{i=1}^kv_{b_i}(G_{b_i}')$. However, this is a contradiction to Proposition~\ref{thm:Lorenz}.

    If $\ell>k$, we have $v_{a_i}(G_{a_i},x_{a_i})=\phi(A,\val)=v_{a_k}(G_{a_k},x_{a_k})<v_{b_k}(G_{b_k}',x_{b_k}')\leq v_{b_i}(G_{b_i}',x_{b_i'})$ for each $i=k+1,\ldots,\ell$, where the only strict inequality in the middle is due to our assumption that $k$ is the smallest index.
    This implies Equation (\ref{eqn:lorenz}) continues to hold if we summing over the first $\ell$ terms instead of the first $k$ terms:
    \begin{equation}\label{eqn:lorenzl}
        \sum_{i=1}^\ell v_{a_i}(G_{a_i},x_{a_i})<\sum_{i=1}^\ell v_{b_i}(G_{b_i}',x_{b_i}').
    \end{equation}
    Since $\ell>k>\ell'$, we have $\sum_{i=1}^\ell v_{a_i}(G_{a_i},x_{a_i})=u+\sum_{i=1}^\ell v_{a_i}(G_{a_i})$ and $\sum_{i=1}^\ell v_{b_i}(G_{b_i}',x_{b_i}')=u+\sum_{i=1}^\ell v_{b_i}(G_{b_i}')$.
    These, together with Equation (\ref{eqn:lorenzl}), give a contradiction to Proposition~\ref{thm:Lorenz} again.
\end{proof}

\iftrue
In the previous proof, we assume $A'$ satisfying the water-filling property w.l.o.g., which is due to the following lemma.
\begin{lemma}
    For any allocation $(G_1,\ldots,G_n)$ on only indivisible goods, if we need to further allocate the single divisible good that values the same to all agents, the full allocation $A$ which applies the water-filling process Lorenz dominates any other full allocation $B$.
\end{lemma}
\begin{proof}
    Since $A$ and $B$ are based on the same allocation of the indivisible goods and all agents have identical valuations on the single divisible good (which can be viewed as water), $A$ Lorenz dominates $B$ if and only if for any $k\in[n]$, the sum of the $k$ largest values of $v_1(A_1),\ldots,v_n(A_n)$ is weakly smaller than the sum of the $k$ largest value of $v_1(B_1),\ldots,v_n(B_n)$.

    Fix a $k\in[n]$. If no water (divisible good) is filled to the $k$ largest values in $A$, then the sum of the $k$ largest values of $v_1(A_1),\ldots,v_n(A_n)$ is weakly smaller than the sum of the $k$ largest value of $v_1(B_1),\ldots,v_n(B_n)$, because they are from the same allocation of the indivisible goods.
    If only the largest $k'<k$ values among the $k$ largest values in $A$ are not filled with some water, the one filled with water is equal to the smallest value in $A$, which is equal to the average of sum of the $n-k'$ smallest values of $v_1(A_1),\ldots,v_n(A_n)$ and the whole divisible good over $n-k'$ agents.
    Then, we can compare the values filled with the water in $A$ ($k-k'$ values in total) with the sum of the following two terms in $B$: The $k-k'$ smallest values among the $k$ largest values, and the extra amount in the largest $k'$ values beyond the $k'$ largest values of $v_1(G_1),\ldots,v_n(G_n)$.
    The sum of the former is obviously weakly smaller than the latter.
    Thus, we prove the statement.
    \end{proof}
\fi

Lorenz domination leads to the following lemma.

\begin{lemma}
Mechanism~\ref{alg:oneidentical}  always outputs allocations that are both leximin and MNW.
\end{lemma}

\iftrue
The leximin part in the above lemma is trivial, while the proof for the MNW part is straightforward from the following lemma.
\begin{lemma}
    For two sorted (from smaller to larger) sequences $A$ and $B$ of length $n$ with all positive entries, if $\sum_{i=1}^k A_i\geq \sum_{i=1}^k B_i$ holds for all $k\in[n]$, then $\prod_{i\in[n]} A_i\geq \prod_{i\in[n]} B_i$.
\end{lemma}
\begin{proof}
This is obviously satisfied when $n=1$. Next, we assume $n\geq 2$.

We define $\Delta:=a_1-b_1$. If $\Delta=0$, the subsequences $A_{2,\ldots,n}$ and $B_{2,\ldots,n}$ satisfy the conditions of this lemma and by reduction, we have $\prod_{i\in[n]} A_i=A_1\times \prod_{i=2}^n A_i\geq B_1\times \prod_{i=2}^n B_i=\prod_{i\in[n]} B_i$.

We assume $\Delta>0$. If we can find an updated subsequence $A'$ obtained from $A$ by decreasing $A_1$ by $\Delta$ and adding $\Delta$ to some numbers among $\{A_2,\ldots,A_n\}$, such that it is ordered from smaller to larger and $\sum_{i=2}^k A_i\geq \sum_{i=2}^k B_i$ holds for all $k\in\{2,\ldots,n\}$, by induction and $A'_1=B_1$, we get $\prod_{i\in[n]} A'_i\geq \prod_{i\in[n]} B_i$. Since the transformation from $A$ to $A'$ can only decrease the product, we have $\prod_{i\in[n]} A_i\geq \prod_{i\in[n]} B_i$.

\begin{algorithm}[t]
    \floatname{algorithm}{Algorithm}
    \caption{An algorithm for transforming a sorted sequence $A$ to a required sequence $A'$ according to a sequence $B$}
    \label{alg:mnw}
    \begin{algorithmic}[1]
        \STATE initialize $A'_i\leftarrow A_i$ for all $i\in[n]$
        \STATE \textbf{for} $k$ from $2$ to $n$ \textbf{do}
        \STATE \hspace{0.1cm} \textbf{if} $\sum_{i=2}^{k} A'_i < \sum_{i=2}^k B_i$ \textbf{then}
        \STATE \hspace{0.3cm} let $\delta=\sum_{i=2}^k B_i-\sum_{i=2}^{k} A'_i$ \label{algmnwforbegin}
        \STATE \hspace{0.3cm} $A'_k\leftarrow A'_k + \delta$
        \STATE \hspace{0.3cm} $A'_1\leftarrow A'_1 - \delta$ \label{algmnwforend}
        \STATE $A'_n\leftarrow A'_n+(A'_1-B_1)$
        \STATE $A'_1\leftarrow B_1$
        \STATE \textbf{return} $A'$
    \end{algorithmic}
\end{algorithm}
It suffices to show we can find such subsequence $A'$.
We can perform the transformation by Algorithm \ref{alg:mnw}.
We now prove the correctness of this algorithm by the induction of $k$. For simplicity, when the variable for the loop is enumerated to the number $k$, we call this Round $k$.
At the beginning, we have $A'_1\geq B_1$ from $A_1\geq B_1$, and $\sum_{i=1}^k A'_i\geq \sum_{i=1}^k B_i$ for all $k\geq 2$. For the monotonicity, $A'_1\leq A'_2$ can always hold since $A'_1$ can only be decreased from $A_1$ and $A'_2$ can only be increased from $A_2$.
We then start the induction.
We assume the first round such that  Steps \ref{algmnwforbegin}-\ref{algmnwforend} will be performed in the following is Round $k$, i.e, $\sum_{i=2}^{k} A'_i < \sum_{i=2}^k B_i$.
Because we have $\sum_{i=1}^k A'_i\geq \sum_{i=1}^k B_i$, $A'_1$ is still weakly larger than $B_1$ after performing Steps \ref{algmnwforbegin}-\ref{algmnwforend}, and $\sum_{i=2}^{k} A'_i$ will be equal to $\sum_{i=2}^k B_i$.
Next, we will show if $k<n$, one of the following two cases will occur.

The first case is that $\sum_{i=2}^{k+1} A'_i \geq \sum_{i=2}^{k+1} B_i$ at this round (Round $k$).
Since we perform Steps \ref{algmnwforbegin}-\ref{algmnwforend} at this round, we have $\sum_{i=2}^{k} A'_i = \sum_{i=2}^k B_i$.
If $k=2$, we have $A'_k=B_k$. If $k>2$, from $\sum_{i=2}^{k-1} A'_i \geq \sum_{i=2}^{k-1} B_i$ (at this round), we have $A'_k\leq B_k$.
Because $B$ is sorted and $\sum_{i=2}^{k+1} A'_i \geq \sum_{i=2}^{k+1} B_i$, we have $A'_k\leq B_k\leq B_{k+1}\leq A'_{k+1}$.

The second case is that $\sum_{i=2}^{k+1} A'_i <\sum_{i=2}^{k+1} B_i$ at this round (Round $k$).
In this case, we will perform  Steps \ref{algmnwforbegin}-\ref{algmnwforend} at next round (Round $k+1$), and we will finally have $\sum_{i=2}^{k+1} A'_i \geq \sum_{i=2}^{k+1} B_i$ (actually the equation is satisfied). From the argument of the first case, we also have $A'_k\leq A'_{k+1}$ at Round $k+1$. Thus, the monotonicity can finally hold between $A'_k$ and $A'_{k+1}$.

Further, since we only transfer the amount between the first $k$ value of sequence $A'$ before and at this round $k$, the inequality $\sum_{i=1}^{k'} A'_i\geq \sum_{i=1}^{k'} B_i$ can still hold for all $k'\geq k$. Thus, we can keep this induction and finally get the satisfied $A'$ after the loop.

For the last two steps of our algorithm, since $A'_1\geq B_1$ can be maintained during the loop, this change won't affect the required conditions of $A'$. Thus, we can finally get the $A'$ we need.
\end{proof}
\fi

\section{Binary Valuations for Both Divisible and Indivisible Goods}
\label{sec:binary-IND-and-DIV}

In this section, we focus on the setting where agents' valuations on both divisible and indivisible goods are binary. We start from the simplest case where there are only two agents.
We first show that no deterministic mechanism can ensure truthfulness and always output an MNW/leximin allocation. We illustrate this by Example~\ref{example:MNWfails}.
This is in sharp contrast to Halpern et al.'s mechanism $\MNW$ and Mechanism~\ref{alg:oneidentical}.

\begin{example}\label{example:MNWfails}
Consider the example shown below.
\begin{table}[H]
\centering
\begin{tabular}{@{}*{4}{c}@{}}
\toprule
& $g_1$ & $d_1$ & $d_2$\\
\midrule
$v_1$ & $1$ & $1$ & $0$ \\
$v_2$ & $1$ & $0$ & $1$ \\
\bottomrule
\end{tabular}
\end{table}

In the allocation returned by a mechanism that can always output an MNW/leximin allocation, each agent $i$ receives the corresponding divisible good $d_i$ and one of them also receives the indivisible good $g_1$.
Without loss of generality, we assume agent $1$ receives $g_1$ and $d_1$.
If the actual valuation of agent $1$ is positive towards all three goods and she reports her valuation truthfully, she will only receive $d_1$ and a half of $d_2$ to achieve an MNW/leximin allocation. Thus, she has an incentive to misreport the valuation we stated in the table to earn a benefit.

\end{example}

\iftrue
Even when allowing randomized mechanisms, we can also show that there is no such truthful (randomized) mechanism that can always output an MNW/leximin allocation with a general number of agents.
\begin{example}
    Consider the following example. We need to allocate five indivisible goods and one divisible good to four agents. The specific valuation profile is listed in the table.
    \begin{table}[H]
        \centering
        \begin{tabular}{@{}*{7}{c}@{}}
        \toprule
        & $g_1$ & $g_2$ & $g_3$ & $g_4$ & $g_5$ & $d_1$\\
        \midrule
        $v_1$ & $0$ & $1$ & $0$ & $1$ & $0$ & $0$\\
        $v_2$ & $1$ & $1$ & $0$ & $0$ & $1$ & $1$ \\
        $v_3$ & $1$ & $1$ & $1$ & $1$ & $1$ & $1$ \\
        $v_4$ & $1$ & $1$ & $1$ & $1$ & $1$ & $1$ \\
        \bottomrule
        \end{tabular}
        \end{table}
    We first consider the MNW allocations in this instance. If agent $1$ only receives one good, the remaining three agents should share the remaining $5$ goods, so the maximum Nash welfare is upper bounded by $1\times \left(\frac{5}{3}\right)^3=\frac{125}{27}$. But if we allocate both $g_2$ and $g_4$ to agent $1$, there are $4$ remaining goods can be allocated to the remaining $3$ agents, and the largest possible Nash welfare is $2\times  \left(\frac{4}{3}\right)^3=\frac{128}{27}$, which is larger than $\frac{125}{27}$. Furthermore, such an allocation exists, which can be achieved by dividing the divisible good $d_1$ into three equal parts, and each remaining agent receives one valuable indivisible good and one (equal) part of $d_1$. Thus, agent $2$ will always value her own bundle $\frac{4}{3}$ in an MNW allocation.

    Next we can consider the case where agent $2$ misreports her valuation towards $g_1$ and $g_5$ from $1$ to $0$.
    If agent $1$ receives both $g_2$ and $g_4$, from the MNW property, $g_1$, $g_3$ and $g_5$ should be allocated to agent $3$ and agent $4$. To achieve the maximum Nash welfare, agent $2$ should receive the whole divisible good $d_1$. The corresponding Nash welfare is $2\times 1 \times 2 \times 1=4$.
    If agent $1$ receives only $1$ good, w.l.o.g., we assume she receives $g_4$.
    Agents $3$ and $4$ should share $g_1$, $g_3$ and $g_5$ according to the MNW property. If one of them receives all three indivisible goods, the Nash welfare is upper bounded by $1\times 1 \times 1 \times 3=3$.
    We then, w.l.o.g., assume agent $3$ receives two of them, and agent $4$ receives the remaining one. The present values of each agent towards her own bundle are $(1,0,2,1)$ and we need to allocate $g_2$ and $d_1$ further.

    If we do not allocate $g_2$ to agent $2$, Nash welfare is upper bounded by $1\times 1 \times 2 \times 2=4$. So we allocate $g_2$ to agent $2$.
    To maximize the Nash welfare, the divisible good $d_1$ should be evenly allocated to agent $2$ and agent $4$, and the corresponding Nash welfare is $1\times 1.5 \times 2 \times 1.5=4.5$, which is the largest.
    After misreporting, agent $2$ will receive the value of $1.5$ in all MNW allocations. Thus, there is no truthful mechanism that can always output an MNW allocation.

    For the leximin allocations in this instance, we can use a similar argument to get the same set of possible allocations before and after the same misreporting from agent $2$. So we can conclude the same result.
\end{example}
\fi

Such an example shows that directly returning an MNW allocation cannot work for our setting.
Instead, we adopt the mechanism $\MNW$ (Theorem~\ref{thm:MNWtruthful}) for only indivisible goods and design a truthful mechanism which can always output an \EFM allocation when there are only two agents.

\begin{algorithm}[t]
    \floatname{algorithm}{Mechanism}
    \caption{A truthful \EFM mechanism for binary valuations on both divisible and indivisible goods with two agents}
    \label{alg:binarytwoagents}
    \begin{algorithmic}[1]
        \STATE use $\MNW$ (Theorem~\ref{thm:MNWtruthful}) to compute an MNW/leximin allocation $(G_1, G_2)$ of $G$
        \STATE initialize $\x_i\leftarrow \mathbf{0}$ for each agent $i\in\{1,2\}$
        \STATE \textbf{if} $v_1(G_1)\neq v_2(G_2)$ \textbf{then}
        \STATE \hspace{0.1cm} let $i^*=\arg\min_{i\in\{1,2\}}\{v_i(G_i)\}$ \label{alg2step1begin}
        \STATE \hspace{0.1cm} \textbf{if} $\exists \bar{k}^*\in[\bar{m}]$ s.t. $v_{i^*}(d_{\bar{k}^*})=1$ \textbf{then}
        \STATE \hspace{0.3cm} $x_{i^*\bar{k}^*}\leftarrow 1$ \label{alg2step1end}
        \STATE \textbf{for} each $\bar{k}=1$ to $\bar{m}$ with $\bar{k}\neq \bar{k}^*$   \textbf{do} \label{alg2step2begin}
        \STATE \hspace{0.1cm} \textbf{if} $v_1(d_{\bar{k}})=v_2(d_{\bar{k}})=1$ \textbf{then}
        \STATE \hspace{0.3cm} $x_{1\bar{k}}\leftarrow 0.5$, $x_{2\bar{k}}\leftarrow 0.5$
        \STATE \hspace{0.1cm} \textbf{else}
        \STATE \hspace{0.3cm} {\bf if} $v_1(d_{\bar{k}})=1$ {\bf then} $x_{1\bar{k}}\leftarrow 1$
        \STATE \hspace{0.3cm} {\bf if} $v_2(d_{\bar{k}})=1$ {\bf then} $x_{2\bar{k}}\leftarrow 1$
        \STATE \textbf{return} allocation $((G_1,\x_1),(G_2,\x_2))$
    \end{algorithmic}
\end{algorithm}

Our mechanism is shown in Mechanism \ref{alg:binarytwoagents}. We first allocate indivisible goods according to $\MNW$ which is truthful and can output an EF1 and leximin allocation. If the valuations of two agents over their own bundles are different, we can allocate one valuable divisible good to the one with a smaller value to eliminate the possible envy, if such a divisible good exists. For each remaining divisible good, we allocate it evenly to all the agents who value it positively. The following theorem shows Mechanism~\ref{alg:binarytwoagents} is truthful and \EFM.

\begin{theorem}
    For allocating both divisible and indivisible goods to two agents with binary valuations, there exists a truthful mechanism that always outputs an \EFM allocation.
\end{theorem}


\begin{proof}
We first show the output allocation  $\A=(A_1,A_2)$ is \EFM. Since the allocation $(G_1,G_2)$ output by $\MNW$ is EF1, so it is also \EFM. If $v_1(G_1)=v_2(G_2)$, such an allocation is envy-free. If $v_1(G_1)\neq v_2(G_2)$, after performing Steps \ref{alg2step1begin}-\ref{alg2step1end}, either the present allocation is envy-free, or $v_{i^*}(D)=0$, which means the output allocation is always \EFM for agent $v_{i^*}$ no matter how divisible goods are allocated.
The remaining steps are just to allocate each divisible good evenly to the agents who positively value them, which will not destroy the envy-freeness. Thus, the final allocation is \EFM.

We next show this mechanism is truthful. We assume we receive the allocation $\A=(A_1,A_2)$ given the true valuation profile $(v_1,v_2)$. We further assume we get the allocation $\A'=(A'_1,A'_2)$ when agent $1$ misreports her valuation to $v'_1$.

In our mechanism, because the allocation of divisible goods after Step \ref{alg2step2begin} is to evenly allocate all goods to the agents who value them, misreporting can make no benefit in this part. For Steps \ref{alg2step1begin}-\ref{alg2step1end}, the decision depends only on the valuations of the indivisible goods, there is also no gain from misreporting the valuations on divisible goods. So the only possible way to benefit is to misreport the valuations on indivisible goods.

Since $\MNW$ is truthful, we have $v_1(G_1)\geq v_1(G'_1)$. We consider the first case where $v_1(G_1)\geq v_1(G'_1)+1$. By misreporting the valuation, the largest benefit is from the change of $i^*$ from agent $2$ to agent $1$ during Steps \ref{alg2step1begin}-\ref{alg2step1end}. Such benefit is upper bounded by $2\times 0.5=1$, because, if agent $1$ values the allocated good positively, she will receive half of this good if such a good is allocated after Step \ref{alg2step2begin}. Thus, there is no need to misreport in this case.

For the second case where $v_1(G_1)=v_1(G'_1)$, we first have $v'_1(G'_1)\geq v_1(G'_1)$ because of the leximin allocation with binary valuations. The benefit of the misreporting is also from the change of $i^*$ during Steps \ref{alg2step1begin}-\ref{alg2step1end}.
If $i^*$ is $2$ under $(v_1,v_2)$ and does not exist under $(v'_1,v_2)$, then $v_2(G'_2)=v'_1(G'_1)\geq v_1(G'_1)=v_1(G_1)>v_2(G_2)$, this violates the leximin property of $(G_1,G_2)$ under $(v_1,v_2)$.
If $i^*$ is $2$ under $(v_1,v_2)$ and $1$ under $(v'_1,v_2)$, the leximin property of $(G_1,G_2)$ under $(v_1,v_2)$ is also violated from the similar argument.
Similar analysis also holds for the case where $i^*$ does not exist under $(v_1,v_2)$ and is $1$ under $(v'_1,v_2)$.
Since all three cases which can benefit agent $1$ cannot occur, misreporting the valuations has no benefit. Thus, this mechanism is truthful.
\end{proof}

\begin{algorithm}[t]
    \floatname{algorithm}{Mechanism}
    \caption{A truthful \EFM mechanism for binary valuations on indivisible goods and a single divisible good}
    \label{alg:binaryonediv}
    \begin{algorithmic}[1]
        \STATE use $\MNW$ (Theorem~\ref{thm:MNWtruthful}) to compute an MNW/leximin allocation $(G_1,\ldots,G_n)$ of $G$
        \STATE initialize $\x_i\leftarrow \mathbf{0}$ for each agent $i\in[n]$
        \STATE let $T=\arg\min_{i\in[n]}\{v_i(G_i)\}$ \label{alg3deft}
        \STATE \textbf{for} each $i\in T$ \textbf{do}
        \STATE \hspace{0.1cm} $x_{i1}\leftarrow \frac{1}{|T|}$
        \STATE \textbf{return} allocation $((G_1,\x_1),\ldots,(G_n,\x_n))$
    \end{algorithmic}
\end{algorithm}

If we consider the setting with more than two agents, we can also design an \EFM and truthful mechanism for allocating indivisible goods and a single divisible good, which is shown in Mechanism \ref{alg:binaryonediv}.
It is worth noting that, in this setting, agents value the divisible good at either~$1$ or~$0$, which is different from the setting in \Cref{sec:binary-IND+money} where all agents have an identical value over the single divisible good.
In this mechanism, after allocating all indivisible goods according to $\MNW$, we allocate the only divisible good evenly to all agents with the smallest $v_i(G_i)$.

\begin{theorem}
    For allocating indivisible goods and a single divisible good with binary valuations, there exists a truthful mechanism that always outputs an \EFMzero (thus \EFM) allocation.
\end{theorem}

\begin{proof}
Since $\MNW$ returns an EF1 allocation and the remaining of our mechanism is just to allocate the single divisible good evenly to agents with the smallest $v_i(G_i)$, no agent in $T$ will envy others in $T$.
Because there is only one divisible good, no agent in $[n]\setminus T$ will envy the agents in $T$ after allocating this good. Thus, the output allocation is \EFMzero.

We then prove this mechanism is truthful. We assume we receive the allocation $\A=(A_1,\ldots,A_n)$ given the true valuation profile $(v_1,\ldots,v_n)$. We further assume we get the allocation $\A'=(A'_1,\ldots,A'_n)$ when agent $1$ misreports her valuation to $v'_1$.

From the truthfulness of $\MNW$, we have $v_1(G_1)\geq v_1(G'_1)$. Since there is only one divisible good, there is no incentive to misreport if $v_1(G_1)\geq v_1(G'_1)+1$.
We then consider the case where $v_1(G_1)=v_1(G'_1)$. If agent $1$ can benefit from misreporting, this means that after misreporting, agent $1$ is in the set $T$ at Step \ref{alg3deft}, and when agent $1$ truthfully reports $v_1$, either she is not in the set $T$ or the size of the set $T$ is larger than that after misreporting.

Since $\MNW$ is leximin, under binary valuations, we have $v'_1(G'_1)\geq v_1(G'_1)=v_1(G_1)$. Because, after misreporting the valuation, agent $1$ is in the set $T$ which contains all agents with the smallest $v_i(G_i)$. Both two cases mentioned above violate the leximin property of the allocation $(G_1,\ldots,G_n)$ of $G$ under the valuation $(v_1,\ldots,v_n)$.  Thus, this mechanism is truthful.
\end{proof}

\begin{remark}
    This mechanism is no longer \EFMzero and \EFM when there is more than one divisible good, because some envy may occur from the agent outside the set $T$ to the agent in $T$ after allocating multiple divisible goods.
\end{remark}

\section{Conclusion}

In this paper, we have studied truthful and fair (i.e., EFM) mechanisms when allocating mixed divisible and indivisible goods.
Our strong impossibility result shows that truthfulness and EFM are incompatible even if there are only two agents and two goods.
We then designed truthful and EFM mechanisms in various restricted settings.

In future research, it would be intriguing to completely determine the compatibility between truthfulness and EFM when agents have binary valuations over all goods.
Given by our impossibility result, another future direction is to consider weaker notions of truthfulness, e.g., \emph{maximin strategy-proofness}~\citep{brams2006better}, \emph{not obviously manipulation (NOM)}~\citep{troyan2020obvious,ortega2022obvious}, and \emph{risk-averse truthfulness}~\citep{BuSoTa23}.

\clearpage
\section*{Acknowledgments}
The research of Biaoshuai Tao was supported by the National Natural Science Foundation of China (Grant No. 62102252).
The research of Shengxin Liu was partially supported by the National Natural Science Foundation of China (Grant No. 62102117), by the Shenzhen Science and Technology Program (Grant No. RCBS20210609103900003), and by the Guangdong Basic and Applied Basic Research Foundation (Grant No. 2023A1515011188), and was partially sponsored by CCF-Huawei Populus Grove Fund (Grant No. CCF-HuaweiLK2022005).
Xinhang Lu was partially supported by ARC Laureate Project FL200100204 on ``Trustworthy AI''.

\bibliographystyle{plainnat}
\bibliography{bibliography}
\end{document}